%% file: arxiv.tex
\documentclass[a4paper,UKenglish,cleveref, autoref, thm-restate]{lipics-v2021}

\pdfoutput=1 
\hideLIPIcs  

\usepackage{xcolor}
\usepackage{tikz}
\usepackage{tikz-network}

\theoremstyle{theorem}
\newtheorem{openquestion}{Open Question}

\theoremstyle{definition}
\newtheorem{construction}{Construction}

\newcommand{\threefield}[3]{$#1\mid#2\mid#3$}
\newcommand{\ProbOne}{\textsc{Interval Scheduling on Eligible Machines}}
\newcommand{\ProbTwo}{\textsc{Interval Scheduling on Unrelated Machines}}
\newcommand{\ProbThree}{\textsc{Unweighted Interval Scheduling on Unrelated Machines}}
\newcommand{\ProbFour}{\textsc{Unweighted Interval Scheduling on Eligible Machines}}

\usepackage{etoolbox}
\newcommand{\appsymb}{$\star$}
\newcommand{\appref}[1]{{\hyperref[proof:#1]{\appsymb}}}
\newcommand{\appLink}[1]{{\hyperref[#1]{\appsymb}}}

\title{Hardness of Interval Scheduling on Unrelated Machines} 

\titlerunning{Hardness of Interval Scheduling on Unrelated Machines} 

\author{Danny~Hermelin}{Department of Industrial Engineering and Management, Ben-Gurion~University~of~the~Negev, 
Beer-Sheva, 
Israel}{hermelin@bgu.ac.il}{}{}

\author{Yuval~Itzhaki}{Department of Industrial Engineering and Management, Ben-Gurion~University~of~the~Negev, 
Beer-Sheva, 
Israel}{ityuval@post.bgu.ac.il}{}{}

\author{Hendrik~Molter}{Department of Industrial Engineering and Management, Ben-Gurion~University~of~the~Negev, 
Beer-Sheva, 
Israel}{molterh@post.bgu.ac.il}{}{}

\author{Dvir~Shabtay}{Department of Industrial Engineering and Management, Ben-Gurion~University~of~the~Negev, 
Beer-Sheva, 
Israel}{dvirs@bgu.ac.il}{}{}

\authorrunning{D.\ Hermelin, Y.\ Itzhaki, H.\ Molter, and D.\ Shabtay} 

\Copyright{Danny Hermelin, Yuval Itzhaki, Hendrik Molter, and Dvir Shabtay} 



\category{} 

\relatedversion{} 

\funding{Supported by the ISF, grant No.~1070/20.}

\acknowledgements{}

\nolinenumbers 

\EventEditors{John Q. Open and Joan R. Access}
\EventNoEds{2}
\EventLongTitle{42nd Conference on Very Important Topics (CVIT 2016)}
\EventShortTitle{CVIT 2016}
\EventAcronym{CVIT}
\EventYear{2016}
\EventDate{December 24--27, 2016}
\EventLocation{Little Whinging, United Kingdom}
\EventLogo{}
\SeriesVolume{42}
\ArticleNo{23}

\begin{document}

\maketitle

\begin{abstract}
We provide new (parameterized) computational hardness results for \textsc{Interval Scheduling on Unrelated Machines}. It is a classical scheduling problem motivated from \emph{just-in-time} or \emph{lean} manufacturing, where the goal is to complete jobs exactly at their deadline. We are given $n$ jobs and $m$ machines. Each job has a deadline, a weight, and a processing time that may be different on each machine. The goal is find a schedule that maximizes the total weight of jobs completed exactly at their deadline. Note that this uniquely defines a processing time interval for each job on each machine.

\textsc{Interval Scheduling on Unrelated Machines} is closely related to coloring interval graphs and has been thoroughly studied for several decades. However, as pointed out by Mnich and van Bevern [Computers \& Operations Research, 2018], the parameterized complexity for the number $m$ of machines as a parameter remained open. We resolve this by showing that \textsc{Interval Scheduling on Unrelated Machines} is W[1]-hard when parameterized by the number $m$ of machines. To this end, we prove W[1]-hardness with respect to $m$ of the special case where we have parallel machines with eligible machine sets for jobs. This answers Open Problem 8 of Mnich and van Bevern's list of 15 open problems in the parameterized complexity of scheduling [Computers \& Operations Research,~2018].

Furthermore, we resolve the computational complexity status of the unweighted version of \textsc{Interval Scheduling on Unrelated Machines} by proving that it is NP-complete. This answers an open question by Sung and Vlach [Journal of Scheduling, 2005].
\end{abstract}


\section{Introduction}
In scheduling problems, we wish to assign jobs to machines in order to maximize a certain optimization objective while respecting certain constraints. In many traditional scheduling settings, jobs can be scheduled to start at any point in time and then need a given \emph{processing time} to be completed.
However, in a typical \emph{interval scheduling problem}, each job can be processed only in a fixed time interval, or sometimes in a set of time intervals, that may vary from machine to machine~\cite{kolen2007interval,kovalyov2007fixed}. Many different variations of interval scheduling have been considered and investigated~\cite{angelelli2014optimal,arkin1987scheduling,bentert2019inductive,van2015interval,van2017parameterized,bouzina1996interval,spieksma1999approximability}. 

In interval scheduling in its most basic form, we are given a set of $n$ \emph{jobs} and a set of $m$ identical parallel \emph{machines} that each can process one job at a time. Each jobs has a \emph{processing time}, a \emph{deadline}, and a \emph{weight}, and shall be processed such that it finishes exactly at its deadline. This uniquely defines an interval for each job in which it can be processed. A \emph{schedule} assigns a subset of the jobs to machines. Unassigned jobs are rejected. We call a schedule \emph{feasible} if no two jobs with overlapping processing time intervals are assigned to the same machine. The goal is to find a feasible schedule that maximizes the weighted number of scheduled jobs.\footnote{In the standard three field notation for scheduling problems of Graham~\cite{Graham1969} this problem is sometimes denoted by \threefield{P}{p_j=d_j-r_j}{\sum_j w_jU_j}, or \threefield{P}{}{\sum_j w_jE_j}, or \threefield{P}{}{\text{JIT}}. We give a more formal definition in \cref{sec:0}.} 

This setting corresponds to the concept of \emph{just-in-time (JIT)} or \emph{lean manufacturing} that revolutionized industrial production processes in the 1980s and 1990s~\cite{krafcik1988triumph,ohno1988toyota,shingo1985revolution,womack1997lean,womack1990machine}. Herein, the main goal is to provide and receive goods precisely when they are needed in order to reduce storage costs and wastage. The first implementation of this manufacturing paradigm is attributed to the Japanese automobile company Toyota and is sometimes also called Toyota Production System (TPS)~\cite{ohno1988toyota,shingo1985revolution}. Naturally, just-in-time and interval scheduling in many different variants has received much attention from the research community since the late 1980s until today~\cite{angelelli2014optimal,baker1990sequencing,bentert2019inductive,van2015interval,van2017parameterized,bouzina1996interval,vcepek2005quadratic,hiraishi2002scheduling,lann1996single,leyvand2010just,shabtay2012just,spieksma1999approximability,sung2005maximizing}.

The basic form of interval scheduling as described above is known to be solvable in polynomial time~\cite{arkin1987scheduling,bouzina1996interval,carlisle1995k,vcepek2005quadratic,hiraishi2002scheduling}. It is closely related to the classical problems of finding maximum independent sets in interval graphs and coloring interval graphs~\cite{carlisle1995k,gavril1974intersection,RoseTL76,yannakakis1987maximum}. The jobs of an interval scheduling instance naturally define an interval graph with vertex weights. For example, if there is only one machine, then interval scheduling is equivalent to finding a maximum weight independent set in an interval graph. Coloring an interval graph or, more specifically, computing its chromatic number is equivalent to determining the minimum number of machines necessary to schedule all jobs.

In our work, we investigate several natural generalizations and variants of interval scheduling and answer some longstanding open questions about their (parameterized) computational complexity.

The first problem we consider in this paper is \ProbOne, a natural generalization of the basic interval scheduling problem we introduced earlier. Here, each job additionally has a set of eligible machines and each job can only be assigned to a machine in this set in a feasible schedule. Arkin and Silverberg~\cite{arkin1987scheduling} proved in 1987 that \ProbOne\ is strongly NP-hard and can be solved in $O(mn^{m+1})$ time. In terms of parameterized complexity, Arkin and Silverberg~\cite{arkin1987scheduling} showed that \ProbOne\ is in XP when parameterized by the number $m$ of machines. However, they left open whether \ProbOne\ also admits an FPT-algorithm for parameter $m$. 
Mnich and van Bevern~\cite{mnich2018parameterized} included this question as Open Problem~8 in their 2018 list of 15 open problems in the parameterized complexity of scheduling. We answer this question negatively in our first main contribution of this paper.

\begin{theorem}\label{thm:w1hard}
\ProbOne\ is strongly W[1]-hard when parameterized by the number $m$ of machines.
\end{theorem}

A natural and well-studied generalization of \ProbOne\ is \ProbTwo. 
In the latter, the processing time of each job can be machine-dependent whereas the deadline stays the same on all machines.
Furthermore, each jobs is eligible on all machines. This definition stems from the just-in-time motivation, where each job should be finished exactly at its deadline but on different machines it may take different times to complete the job. We mention in passing that if both processing times and deadlines can be machine-dependent, the problem becomes NP-hard on two machines~\cite{kovalyov2007fixed,spieksma1999approximability}.
Sung and Vlach~\cite{sung2005maximizing} showed that \ProbTwo\ can also be solved in $O(mn^{m+1})$ time, generalizing the result of Arkin and Silverberg~\cite{arkin1987scheduling}. 
Mnich and van Bevern~\cite{mnich2018parameterized} asked in Open Problem~8 for an FPT-algorithm for \ProbOne\ parameterized by the number $m$ of machines as a first step towards finding an FPT-algorithm for \ProbTwo\ parameterized by $m$.
However, \cref{thm:w1hard} naturally implies that \ProbTwo\ presumably also does not admit an FPT-algorithm for the number $m$ of machines as a parameter.

\begin{corollary}\label{thm:w1hard2}
\ProbTwo\ is strongly W[1]-hard when parameterized by the number $m$ of machines.
\end{corollary}

We point out that all known hardness reductions for \ProbTwo\ require job weights, raising the question whether the weights play an integral role in the computational complexity of the problem. 
\ProbThree\ is the natural special case of \ProbTwo\ where all jobs have weight one.
Sung and Vlach~\cite{sung2005maximizing} asked in 2005 to resolve the computational complexity status of \ProbThree. We give an answer to this in our second main contribution.

\begin{theorem}\label{thm:nph}
\ProbThree\ is NP-complete.
\end{theorem}

We remark that our reduction for \cref{thm:nph} does not imply hardness for the unweighted version of \ProbOne. We leave this open for future research. An additional immediate question that we leave open for future research is whether \ProbThree\ admits an FPT-algorithm for the number $m$ of machines as a parameter.

With \cref{thm:w1hard}, \cref{thm:w1hard2}, and \cref{thm:nph} we answer fundamental longstanding open questions concerning the (parameterized) computational complexity of natural interval scheduling problems. For \ProbOne\ and \ProbTwo, our results together with the XP-containment results from Arkin and Silverberg~\cite{arkin1987scheduling} and Sung and Vlach~\cite{sung2005maximizing}, respectively, essentially resolve their parameterized complexity classification for the number $m$ of machines as a parameter.
We point out that all considered problem variants are known to be fixed-parameter tractable when parameterized by the number $n$ of jobs. This can be shown with a simple reduction to \textsc{Multicolored Independent Set on Interval Graphs} parameterized by the number of colors, which is known to be fixed-parameter tractable~\cite{bentert2019inductive,van2015interval}. 
Hence, we make an important further step towards fully understanding the parameterized complexity
of several basic and natural interval scheduling problems.
We remark that our results also imply that \textsc{Multicolored Independent Set on Interval Graphs} is W[1]-hard when parameterized by the maximum number of vertices of any color.

The rest of paper is organized as follows: we give formal definitions of all problems in \cref{sec:0}. We prove \cref{thm:w1hard} and \cref{thm:w1hard2} in \cref{sec:1} and we prove \cref{thm:nph} inn \cref{sec:2}. 
We conclude with future research directions in \cref{sec:3}.

\section{Problem Setting}\label{sec:0}

The first problem we consider is \ProbOne.
Here, we have a set of $n$ jobs $\{j_1, j_2,\ldots, j_n\}$ and a set of $m$ machines $\{i_1, i_2, \ldots, i_m\}$ that each can process one job at a time.
Each job $j$ has a \emph{processing time} $p_j$, a \emph{deadline} $d_j$, a \emph{weight} $w_j$, and a set of \emph{eligible machines} $M_j\subseteq \{i_1, i_2, \ldots, i_m\}$. Job $j$ can be processed in exactly \emph{one} fixed time interval $(d_j-p_j,d_j]$, specified by its processing time and deadline, that is the same on each of its eligible machines. 
A schedule is a mapping from jobs to machines. More formally, a \emph{schedule} is a function $\sigma : \{j_1, j_2, \ldots, j_n\}\rightarrow \{i_1, i_2, \ldots, i_m, \bot\}$. If for job $j$ we have $\sigma(j)=i$ (with $i\neq\bot$), then job $j$ is scheduled to be processed on machine $i$. If for job $j$ we have $\sigma(j)=\bot$, then job $j$ is not scheduled, that is, it is not assigned to any machine.
We say that two jobs $j,j'$ are \emph{in conflict} on a machine $i$ if $(d_j-p_j,d_j]\cap(d_{j'}-p_{j'},d_{j'}]\neq\emptyset$, that is, the processing time intervals corresponding to jobs $j$ and $j'$ on machine $i$ overlap.
A schedule $\sigma$ is \emph{feasible} if there is no pair of jobs $j,j'$ with $\sigma(j)=\sigma(j')=i\neq \bot$ that is in conflict on machine $i$ and each job is mapped to one of its eligible machines. 
The goal is to find a feasible schedule that maximizes the weighted number of scheduled jobs $W=\sum_{j\mid \sigma(j)\neq \bot}w_j$.
In the standard three field notation for scheduling problems of Graham~\cite{Graham1969} \ProbOne\ is sometimes denoted by \threefield{P}{M_j,p_j=d_j-r_j}{\sum_j w_jU_j}, or \threefield{P}{M_j}{\sum_j w_jE_j}, or \threefield{P}{M_j}{\text{JIT}}.

The second problem we consider is \ProbTwo.
Here, for each job $j$ the processing time $p_{i,j}$ can depend on machine $i$ whereas the deadline $d_j$ is the same on all machines. 
Hence, the processing time interval of job $j$ on machine $i$ is $(d_j-p_{i,j},d_j]$.
Moreover, the all jobs are eligible on all machines, that is, $M_j=\{i_1, i_2, \ldots, i_m\}$ for all jobs $j$.
In the standard three field notation for scheduling problems of Graham~\cite{Graham1969} \ProbTwo\ is sometimes denoted by \threefield{R}{p_j=d_j-r_j}{\sum_j w_jU_j}, or \threefield{R}{}{\sum_j w_jE_j}, or \threefield{R}{}{\text{JIT}}.

Finally, the third problem we consider is \ProbThree, the unweighted version of \ProbTwo. Here, we have that $w_j=1$ for all jobs $j$. In the standard three field notation for scheduling problems of Graham~\cite{Graham1969} \ProbThree\ is sometimes denoted by \threefield{R}{p_j=d_j-r_j}{\sum_j U_j}, or \threefield{R}{}{\sum_j E_j}, or \threefield{R}{w_j=1}{\text{JIT}}.

\section{W[1]-Hardness of \ProbOne}\label{sec:1}

In this section, we prove \cref{thm:w1hard} from which \cref{thm:w1hard2} follows directly.
To prove \cref{thm:w1hard}, we present a parameterized polynomial-time reduction from \textsc{Multicolored Clique} parameterized by the number of colors to \ProbOne\ parameterized by the number $m$ of machines. In \textsc{Multicolored Clique}, we are given a $k$-partite graph $G=(V_1\uplus V_2\uplus\ldots\uplus V_k,E)$, we are asked whether $G$ contains a clique of size $k$. The $k$ vertex parts $V_1, V_2, \ldots, V_k$ are called \emph{colors}. \textsc{Multicolored Clique} parameterized by $k$ is known to be W[1]-hard~\cite{fellows2009multipleinterval}.

Given an instance of \textsc{Multicolored Clique}, we construct an instance of \ProbOne\ as follows. 

\begin{construction}\label{const:w1hard}
Let $G=(V_1\uplus V_2\uplus\ldots\uplus V_k,E)$ be a $k$-partite graph with $n_G$ vertices.
Assume we have some total ordering $<_\pi$ over $V:=V_1\uplus V_2\uplus\ldots\uplus V_k$ such that for all $v\in V_\ell$ and $w\in V_{\ell'}$ we have that if $\ell<\ell'$ then $v<_\pi w$. Let $\pi(v)$ denote the ordinal position of $v\in V$ in the ordering $<_\pi$.

In the following, we describe the jobs and specify their processing times, deadlines and weights. Then we describe the machines and the eligible machine sets for the jobs. 
In order to describe the weights more easily, we introduce the following three values: $c_1=n_G+1$, $c_2=(k-1)n_Gc_1+n_G+1$, and $c_3=(kn_G+k^2n_G)n_Gc_2+1$. 
We create the following jobs:
\begin{itemize}
    \item For each vertex $v\in V$, we create $k$ \emph{vertex jobs} $j_v^{(1)}, j_v^{(2)},\ldots, j_v^{(k)}$, where one of the vertex jobs corresponds to the color of $v$ and the $k-1$ other vertex jobs correspond to the other $k-1$ colors. 
    
    Let $v\in V_\ell$. The processing time of $j_v^{(\ell)}$ (the vertex job corresponding to the same color as $v$) is $k+2$, the deadline of $j_v^{(\ell)}$ is $(k+2)\pi(v)+1$, and the weight of $j_v^{(\ell)}$ is one. 
    
    The processing time of $j_v^{(\ell')}$ with $\ell'\neq \ell$ (vertex jobs corresponding to a different color then $v$) is one, the deadline of $j_v^{(\ell')}$ with $\ell'\neq \ell$ is $(k+2)\pi(v)-\ell'$, and the weight of $j_v^{(\ell')}$ with $\ell'\neq \ell$ is $c_1$. 
    \item For each edge $e=\{v,w\}\in E$ with $v\in V_\ell$, $w\in V_{\ell'}$, and $\ell<\ell'$, we create one \emph{edge job} $j_e$ with processing time $(k+2)(\pi(w)-\pi(v))-\ell+\ell'$, deadline $(k+2)\pi(w)-\ell-1$, and weight $c_2(\pi(w)-\pi(v))+c_3$.
    \item For each color combination $\ell,\ell'$ with $\ell<\ell'$ we create $|V_\ell|+|V_{\ell'}|$ \emph{color combination jobs}, one for each $v\in V_\ell$ and one for each $w\in V_{\ell'}$. 
    
    Let $v\in V_{\ell}$, we create a job $j_v^{(\ell,\ell')}$ with processing time $(k+2)\pi(v)-\ell'-2$, deadline $(k+2)\pi(v)-\ell'-1$, and weight $c_2\pi(v)$. 
    
    Let $w\in V_{\ell'}$, we create a job $j_w^{(\ell,\ell')}$ with processing time $(k+2)(n_G-\pi(w))+\ell+2$, deadline $(k+2)n_G+2$, and weight $c_2(n_G-\pi(w))$.
\end{itemize}

We create $m=\binom{k}{2}+1$ machines $i_1,i_2,\ldots, i_{\binom{k}{2}+1}$. We call the first $\binom{k}{2}$ machines \emph{edge selection machines} (one machine for each color combination) and we call the remaining machine \emph{validation machine}.

\begin{figure}[!t]
    \centering
    \begin{tikzpicture}[scale=0.75,yscale=1.5]
        \input{figures/W1_edge_selection}
    \end{tikzpicture}
    \caption{Illustration of the edge selection machine for color combination $\ell,\ell'$ with $\ell<\ell'$. Depicted are intervals of jobs relating to $v\in V_\ell$, $w\in V_{\ell'}$, and $e=\{v,w\}\in E$. Gray intervals correspond to jobs that are not eligible on the machine.}
    \label{fig:w1_edge_selection}
\end{figure}
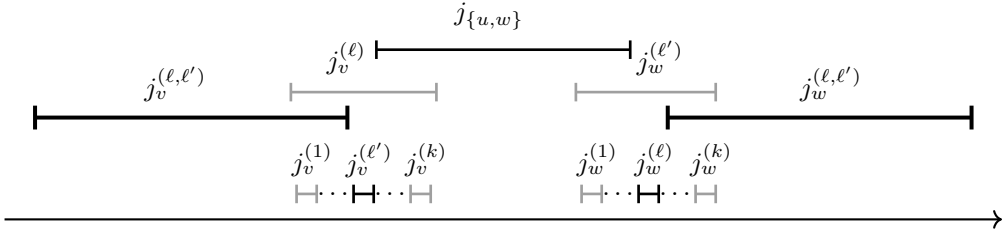

Consider color combination $\ell,\ell'$ with $\ell<\ell'$ and let $i$ be the corresponding edge selection machine.
\begin{itemize}
    \item For each vertex $v\in V_\ell$ we add machine $i$ to the set of eligible machines of job $j_v^{(\ell')}$ and of job $j_v^{(\ell,\ell')}$.
    \item For each vertex $w\in V_{\ell'}$ we add machine $i$ to the set of eligible machines of job $j_w^{(\ell)}$ and of job $j_w^{(\ell,\ell')}$.
    \item For each edge $e=\{v,w\}\in E$ with $v\in V_\ell$ and $w\in V_{\ell'}$ we add machine $i$ to the set of eligible machines of job $j_e$.
\end{itemize}
We give an illustration of the edge selection machines in \cref{fig:w1_edge_selection}.
Finally, consider the validation machine $i_{\binom{k}{2}+1}$. We add the validation machine to the set of eligible machines of all vertex jobs.
We give an illustration of the validation machine in \cref{fig:w1_validation_machine}.

\begin{figure}[!t]
    \centering
    \begin{tikzpicture}[scale=0.85,yscale=1.1]
        \input{figures/W1_validation_machine}
    \end{tikzpicture}
    \caption{Illustration of the validation machine. Depicted are intervals of jobs corresponding to vertices $v,u,w\in V_\ell$ with $v<_\pi u<_\pi w$.}
    \label{fig:w1_validation_machine}
\end{figure}
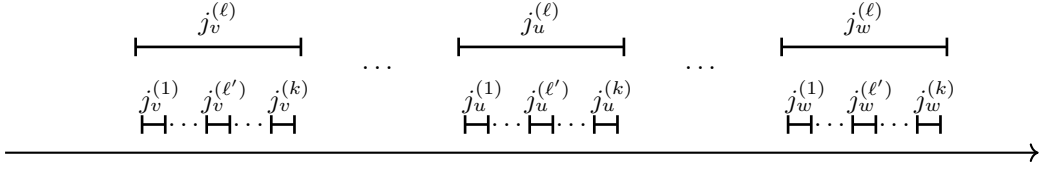
\end{construction}

This finishes our construction of the \ProbOne\ instance. We first show that given a clique of size $k$ in $G$, we can create a feasible schedule for the constructed instance such that the total weight of scheduled jobs attains at least a certain value.

\begin{lemma}\label{lem:corr1}
Let $G$ be an instance of \textsc{Multicolored Clique}. Let $I$ be the \ProbOne\ instance computed from $G$ as specified by \cref{const:w1hard}. If $G$ contains a clique of size $k$, then there is a feasible schedule $\sigma$ for $I$ such that for the total weight $W$ of scheduled jobs we have
\[
W\ge\binom{k}{2}c_3 + \binom{k}{2}n_G c_2 + (k-1)n_G c_1 + k.
\]
\end{lemma}
\begin{proof}
Let $G=(V_1\uplus V_2\uplus\ldots\uplus V_k,E)$ be an instance of \textsc{Multicolored Clique} and consider the corresponding \ProbOne\ instance specified by \cref{const:w1hard}. Assume there is a clique $X$ of size $k$ in $G$. Then we schedule the following jobs.
\begin{itemize}
    \item For color combination $\ell,\ell'$ with $\ell<\ell'$ let $\{v\}=X\cap V_\ell$ and
$\{w\}=X\cap V_{\ell'}$. Since $X$ is a clique in $G$, we know that $e=\{v,w\}\in E$. On edge selection machine $i$ corresponding to color combination $\ell,\ell'$ we schedule the following jobs: $j_e$, $j_v^{(\ell,\ell')}$, $j_w^{(\ell,\ell')}$, $j_v^{(\ell')}$, and $j_w^{(\ell)}$. By construction of the instance, the intervals of the jobs are non-intersecting on the machine $i$, and machines $i$ is in the eligible set of the four jobs. Hence, scheduling these jobs on machine $i$ yields a feasible schedule. Furthermore, it accounts for weight $c_3+n_G c_2+2c_1$ of scheduled jobs per color combination.

Summing over all color combinations, we obtain weight $\binom{k}{2}c_3 + \binom{k}{2}n_G c_2+k(k-1)c_1$. Note that for each $\{v\}=X\cap V_\ell$ we have scheduled all vertex jobs $j_v^{(\ell')}$ with $\ell\neq \ell'$.
\item On the validation machine $i_{\binom{k}{2}+1}$ we schedule the following jobs. Let $\{v\}=X\cap V_\ell$, then we schedule vertex job~$j_v^{(\ell)}$. Note that this job is only in conflict with vertex jobs $j_v^{(\ell')}$ with $\ell\neq\ell'$, which are scheduled on the edge selection machines.
Furthermore, we schedule all jobs $j_w^{(\ell')}$ with $v\neq w\in V_\ell$ and $\ell\neq \ell'$.
By construction, all these jobs can be scheduled on machine $i_{\binom{k}{2}+1}$ without conflicts and all these jobs have machine $i_{\binom{k}{2}+1}$ in their set of eligible machines. 

For all colors, this accounts for weight $(n_G-k)(k-1)c_1+k$ of scheduled jobs.
\end{itemize}

Clearly, we have that the constructed schedule is feasible. Furthermore, it is straightforward to check that the total weight of scheduled jobs in this constructed schedule is $W$.
\end{proof}

Before we show a similar statement for the opposite direction, we make an observation about feasible schedules in \ProbOne\ instances from \cref{const:w1hard}.
%
%
%
%
We show that we can assume that any feasible schedule where the total weight of scheduled jobs is at least $\binom{k}{2}c_3 + \binom{k}{2}n_G c_2 + (k-1)n_G c_1 + k$ schedules exactly one edge job on each edge selection machine.

\begin{observation}\label{obs:edges}
Let $I$ be an instance of \ProbOne\ resulting from applying \cref{const:w1hard} to some $k$-partite graph $G$. Let $\sigma$ be a feasible schedule such that for the total weight $W$ of scheduled jobs we have
\[
W\ge\binom{k}{2}c_3 + \binom{k}{2}n_G c_2 + (k-1)n_G c_1 + k.
\]
Then exactly $\binom{k}{2}$ edge jobs are scheduled, one on each edge selection machine.
\end{observation}

\begin{proof}
We first show that no feasible schedule with total weight $W\ge\binom{k}{2}c_3 + \binom{k}{2}n_G c_2 + (k-1)n_G c_1 + k$ of scheduled jobs can schedule more than $\binom{k}{2}$ edge jobs. 
Let $\ell,\ell'$ with $\ell<\ell'$ be a color combination. On the edge selection machine corresponding to color combination $\ell,\ell'$ we have that all edge jobs corresponding to edges that do \emph{not} connect vertices of colors $\ell$ and $\ell'$ are not eligible. Furthermore, all edge jobs corresponding to edges that connect vertices of colors $\ell$ and $\ell'$ are pairwise in conflict. It follows that on each edge selection machine, at most one edge job can be scheduled. Hence, any feasible schedule can schedule at most $\binom{k}{2}$ edge jobs, one on each edge selection machine.

We next show that any feasible schedule with $W\ge\binom{k}{2}c_3 + \binom{k}{2}n_G c_2 + (k-1)n_G c_1 + k$ needs to schedule at least $\binom{k}{2}$ edge jobs. Assume we have a feasible schedule that schedules (strictly) less than $\binom{k}{2}$ edge jobs. Note that each edge job has weight at least $c_3$. Furthermore, there are at most $kn_G+k^2n_G$ jobs that are not edge jobs and those jobs each have weight at most $n_Gc_2$.
Let $\sigma$ be a feasible schedule that does not schedule all edge jobs and let $W$ be the total weight of all jobs scheduled by $\sigma$. Let $W^\star$ denote the sum of weights of all jobs that are not edge jobs. Then we have
\[
W < (\binom{k}{2}-1)c_3+W^\star < \binom{k}{2}c_3.
\]
Hence, the observation follows.
\end{proof}

Now we are ready to show how to construct a clique of size $k$ from a feasible schedule where the total weight of scheduled jobs is at least $\binom{k}{2}c_3 + \binom{k}{2}n_G c_2 + (k-1)n_G c_1 + k$.

\begin{lemma}\label{lem:corr2}
Let $G$ be an instance of \textsc{Multicolored Clique}. Let $I$ be the \ProbOne\ instance computed from $G$ as specified by \cref{const:w1hard}. If there is a feasible schedule $\sigma$ for $I$ such that for the total weight $W$ of scheduled jobs we have
\[
W\ge\binom{k}{2}c_3 + \binom{k}{2}n_G c_2 + (k-1)n_G c_1 + k,
\]
then $G$ contains a clique of size $k$.
\end{lemma}

\begin{proof}
Let $G=(V_1\uplus V_2\uplus\ldots\uplus V_k,E)$ be an instance of \textsc{Multicolored Clique} and consider the corresponding \ProbOne\ instance specified by \cref{const:w1hard}. Assume we have a feasible schedule $\sigma$ such that the for the total weight $W$ of scheduled jobs we have
$W\ge\binom{k}{2}c_3 + \binom{k}{2}n_G c_2 + (k-1)n_G c_1 + k$. We construct a clique of size $k$ in $G$ as follows.

By \cref{obs:edges} we know that $\sigma$ schedules one edge job on each edge selection machine. 
We can also observe that on the validation machine, only vertex jobs are eligible and can be scheduled. Note that the sum of weights of all vertex jobs is $(k-1)n_Gc_1+n_G$, which is strictly smaller than $c_2$.

 Assume that edge job $j_e$ is scheduled on the edge selection machine $i$ corresponding to color combination $\ell,\ell'$ with $\ell<\ell'$. Then by construction, $e=\{v,w\}\in E$ with $v\in V_\ell$ and $w\in V_{\ell'}$. Now by construction of the instance, two color combination jobs can be scheduled on machine $i$, one for a vertex of color $\ell$ and one for a vertex of color $\ell'$.
 In order to obtain a weight of scheduled jobs of at least $c_3+n_G c_2$ it is necessary that jobs $j_v^{(\ell,\ell')}$ and $j_w^{(\ell,\ell')}$ are scheduled (note that the weights of $j_e$, $j_v^{(\ell,\ell')}$, and $j_w^{(\ell,\ell')}$ sum up to exactly $c_3+n_G c_2$). Any other selection of color combination jobs to schedule either results in a weight that is lower by at least $c_2$ or in an infeasible schedule. 
 Now, by construction, the only further jobs that can be scheduled are $j_v^{(\ell')}$ and $j_w^{(\ell)}$. It follows that the maximum weight achievable on any edge selection machine is $c_3+n_g c_2+2c_1$. Since $\binom{k}{2}2c_1<c_2$, it follows that for each edge selection machine corresponding to color combination $\ell,\ell'$ with $\ell<\ell'$ we have the following: one edge job $j_e$ for $e=\{v,w\}$ with $v\in V_\ell$ and $w\in V_{\ell'}$ is scheduled and the two color combination jobs $j_v^{(\ell,\ell')}$ and $j_w^{(\ell,\ell')}$ are scheduled.
 
 We can conclude that the jobs scheduled on all edge selection machines have a total weight of at least $\binom{k}{2}c_3 + \binom{k}{2}n_G c_2$. Hence, there are additional jobs scheduled that have a total weight of $(k-1)n_G c_1 + k$ on the validation machine.
 
 Since no additional edge jobs or color combination jobs can be scheduled, we have that all $(k-1)n_G$ vertex jobs $j_v^{(\ell')}$ with $v\in V_\ell$ and $\ell\neq \ell'$ (having weight $c_1>n_G$) are scheduled. Furthermore, at least $k$ vertex jobs $j_v^{(\ell)}$ with $v\in V_\ell$ (having weight one) are scheduled.
 
 Let $X$ be the set of vertices in $G$ such that if $v\in X$ and $v\in V_\ell$, the job $j_v^{(\ell)}$ is scheduled. We claim that $X$ is a clique of size at least $k$ in $G$.
 
 By construction we have that $|X|\ge k$, assume for contradiction that $X$ is not a clique in $G$. Then there are two vertices $v,w\in X$ such that $e=\{v,w\}\notin E$. Let $v\in V_\ell$ and $w\in V_{\ell'}$. Then, in particular, vertex jobs $j_v^{(\ell')}$ and $j_w^{(\ell)}$ cannot be scheduled on the validation machine, since they are in conflict with vertex jobs $j_v^{(\ell)}$ and $j_w^{(\ell')}$, respectively. However, as observed above, vertex jobs $j_v^{(\ell')}$ and $j_w^{(\ell)}$ can only be scheduled on the edge selection machine corresponding to color combination $\ell,\ell'$ if there is edge job $j_e$ with $e=\{v,w\}$ scheduled on that machine, a contradiction to the assumption that $e=\{v,w\}\notin E$.
\end{proof}

Finally, we have all ingredients to prove \cref{thm:w1hard}.

\begin{proof}[Proof of \cref{thm:w1hard}]
To prove \cref{thm:w1hard}, we show that \cref{const:w1hard} is parameterized polynomial-time reduction from \textsc{Multicolored Clique} parameterized by the number of colors to \ProbOne\ parameterized by the number $m$ of machines. First, it is easy to observe that given an instance of \textsc{Multicolored Clique}, the \ProbOne\ instance specified by \cref{const:w1hard} can be computed in polynomial time. Furthermore, if $k$ is the number of colors in the \textsc{Multicolored Clique} instance, then the number of machines in the constructed \ProbOne\ instance is $m=\binom{k}{2}+1$. Lastly, observe that all weights in the constructed \ProbOne\ instance are in $n_G^{O(1)}$ (where $n_G$ is the number of vertices in the \textsc{Multicolored Clique} instance).

The correctness of the reduction follows from \cref{lem:corr1,lem:corr2}. Since \textsc{Multicolored Clique} parameterized by the number of colors is W[1]-hard~\cite{fellows2009multipleinterval}, we have that \cref{thm:w1hard} follows.
\end{proof}

\section{NP-Hardness of \ProbThree}\label{sec:2}
In this section we prove \cref{thm:nph}. The containment of \ProbThree\ in NP is easy to see, hence we focus on proving NP-hardness.
To this end, we present a polynomial-time many-one reduction from \textsc{Exact (3,4)-SAT} to \ProbThree. In \textsc{Exact (3,4)-SAT} we are given a Boolean formula $\phi$ in conjunctive normal form where every clause has exactly three literals and every variable appears in exactly four clauses, and are asked whether $\phi$ has a satisfying assignment. \textsc{Exact (3,4)-SAT} is known to be NP-hard~\cite{tovey1984simplified}.

Given such a formula $\phi$, we construct an instance $I$ of \ProbThree\ as follows.

\begin{construction}\label{const:nph}
Let $\phi$ be a Boolean formula in conjunctive normal form where every clause has exactly three literals and every variable appears in exactly four clauses. Let $\alpha$ be the number of variables in $\phi$ and let~$\beta$ be the number of clauses in $\phi$. We construct an instance $I$ of \ProbThree\ as follows.

We first describe the jobs, then we define an ordering of the jobs and use it to specify their deadlines. Lastly, we describe the processing times of the jobs on the different machines. We create the following jobs.
\begin{itemize}
    \item For every variable $x$, we create two \emph{variable jobs}: $x^T$ and $x^F$.
    \item For every clause $c$, we create three \emph{clause jobs}: $c_{1}$, $c_{2}$, and $c_{3}$.
    \item We create $2\alpha+2\beta$ \emph{dummy jobs}. 
\end{itemize}

We next define an ordering $\pi$ of the jobs, which we will use to define the deadlines of the jobs.
To this end, we partition the jobs into the following sets.
\begin{itemize}
    \item Let $T=\{x^T\mid x\text{ is a variable in }\phi\}$.
    \item Let $F=\{x^F\mid x\text{ is a variable in }\phi\}$.
    \item Let $P=\{c_{\ell}\mid \text{the }\ell\text{th literal of clause }c\text{ of }\phi\text{ is non-negated}\}$.
    \item Let $N=\{c_{\ell}\mid \text{the }\ell\text{th literal of clause }c\text{ of }\phi\text{ is negated}\}$.
    \item Let $D$ be the set of dummy jobs.
\end{itemize}
Now we define $\pi$ as a total ordering of the jobs such that
\[
D\prec N\prec F\prec P\prec T,
\]
and the jobs within the sets are ordered in an arbitrary but fixed way.
Let $\pi(j)$ denote the ordinal position of job $j$ in $\pi$. For each job $j$, we set 
\[
d_j=\pi(j).
\]

We next describe the machines, more specifically, the processing times of all jobs on each of the machines.
We first introduce $\alpha$ \emph{variable selection machines}, one for each variable in $\phi$. Let $x$ be a variable in $\phi$, then we introduce a machine where the processing time of variable job $x^T$ is $\pi(x^T)-2\alpha-2\beta$ and the processing time of variable job $x^F$ is $\pi(x^F)-2\alpha-2\beta$. We set the processing times of all other jobs to their respective deadlines.
We give an illustration of the variable selection machines in \cref{fig:variable_selection}.

\begin{figure}[!t]
    \centering
    \begin{tikzpicture}[scale=0.75,yscale=.75]
        \input{figures/variable_selection}
    \end{tikzpicture}
    \caption{Illustration of the job intervals on the variable selection machine for variable $x$. On this machine only one of jobs $x^T$ and $x^F$ (bold) can be scheduled alongside with one dummy job.}
    \label{fig:variable_selection}
\end{figure}
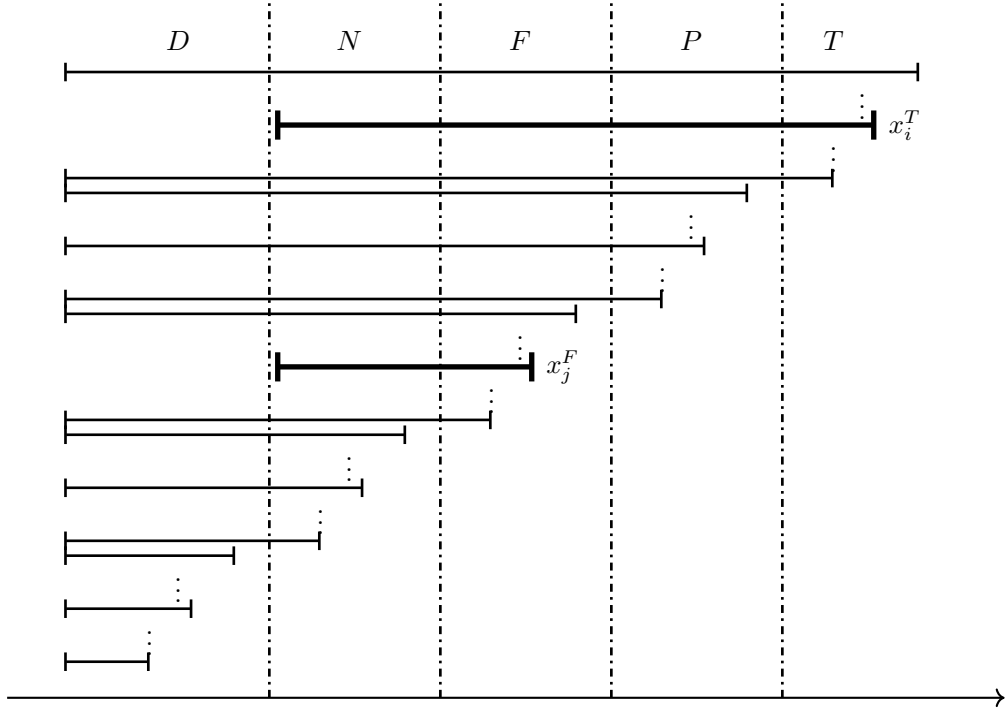

Next, we introduce $2\beta$ \emph{clause selection machines}, two for each clause in~$\phi$. Let $c$ be a clause in~$\phi$, then we introduce two machines where  the processing times of ${c_1}$, ${c_2}$, and ${c_3}$ are $\pi(c_1)-2\alpha-2\beta$, $\pi(c_2)-2\alpha-2\beta$, and $\pi(c_3)-2\alpha-2\beta$, respectively. We set the processing times of all other jobs to their respective deadlines.
We give an illustration of the clause selection machines in \cref{fig:clause_selection}.

\begin{figure}[!t]
    \centering
    \begin{tikzpicture}[scale=0.75,yscale=.75]
        \input{figures/clause_selection}
    \end{tikzpicture}
    \caption{Illustration of the job intervals on the clause selection machine for clause $c$. On this machine only one of the jobs $c_1$, $c_2$ and $c_3$ (bold) can be scheduled alongside with one dummy job.}
    \label{fig:clause_selection}
\end{figure}
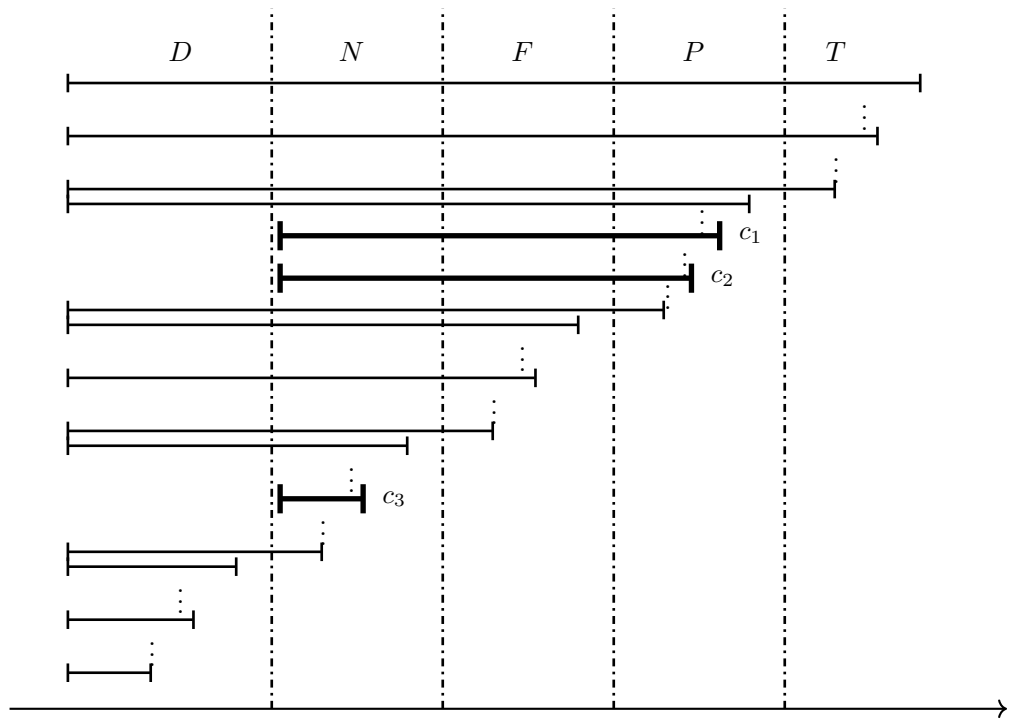

Furthermore, we have $\alpha$ \emph{validation machines}, one for each variable in $\phi$. Let $x$ be a variable in $\phi$, then we introduce a machine where 
\begin{itemize}
    \item the processing time of ${x^T}$ is $\pi(x^T)-\pi(x^F)+1$, 
    \item the processing time of ${x^F}$ is $\pi(x^F)-2\alpha-2\beta$,
    \item if $x$ appears in the $\ell$th literal of clause $c$, then the processing time of ${c_\ell}$ is one, and
    \item processing times of all other jobs are set to their respective deadlines.
\end{itemize}
We give an illustration of the validation machines in \cref{fig:validation_machine}.

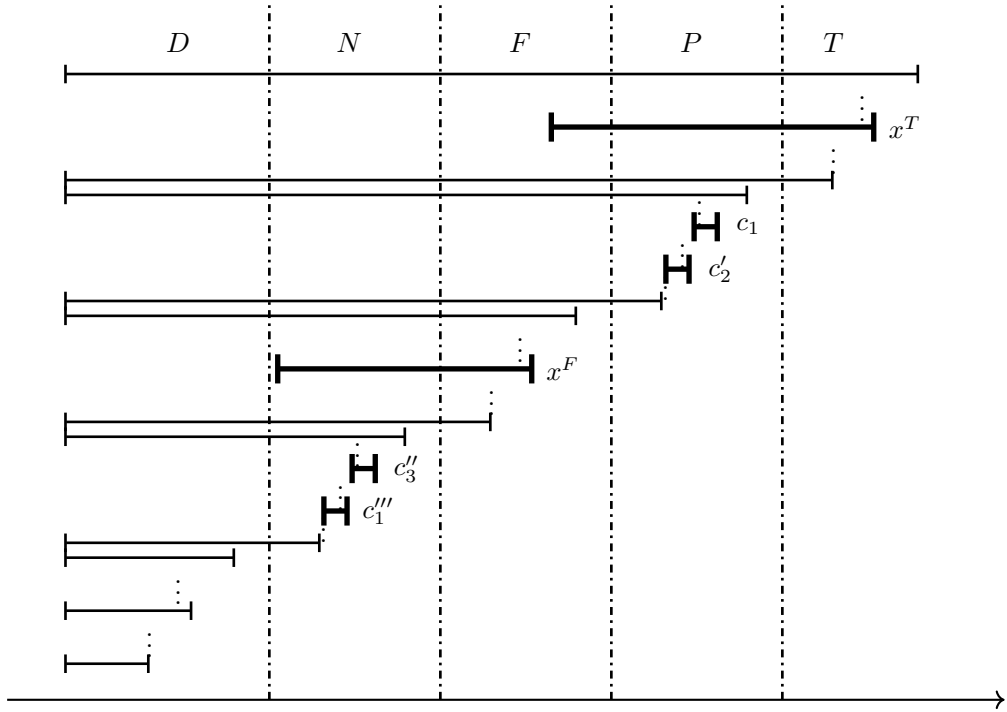
\begin{figure}[!t]
    \centering
    \begin{tikzpicture}[scale=.75,yscale=.75]
        \input{figures/validation_machine}
    \end{tikzpicture}
    \caption{Illustration of the job intervals on the validation machine for variable $x$ for the case that $x$ appears in clauses $c$ and $c'$ non-negated in positions one and two, respectively, and that $x$ appears in clauses $c''$ and $c'''$ negated in positions three and one, respectively. Processing time intervals of jobs that do not conflict with dummy jobs on this machine are depicted in bold.}
    \label{fig:validation_machine}
\end{figure}

\end{construction}

This finishes our construction of the \ProbThree\ instance. 
%
We first show that given a satisfying assignment for $\phi$, we can create a feasible schedule for the constructed instance such all jobs are scheduled.

\begin{lemma}\label{lem:nphcorr1}
Let $\phi$ be an instance of \textsc{Exact (3,4)-SAT}. Let $I$ be the \ProbThree\ instance computed from $\phi$ as specified by \cref{const:nph}. If $\phi$ is satisfiable, then there is a feasible schedule $\sigma$ for $I$ such that all jobs are scheduled.
\end{lemma}
\begin{proof}
Let $\phi$ be an instance of \textsc{Exact (3,4)-SAT} and consider the corresponding \ProbThree\ instance specified by \cref{const:nph}.
Assume there is a satisfying assignment for $\phi$. Then we schedule the jobs as follows.

We first describe on which machine we schedule each variable job. Let $x$ be a variable in $\phi$.
If~$x$ is set to true in the satisfying assignment, we schedule variable job $x^T$ on the variable selection machine corresponding to $x$ and we schedule variable job $x^F$ on the validation machine corresponding to~$x$. Otherwise, we schedule variable job $x^F$ on the variable selection machine corresponding to $x$ and we schedule variable job $x^T$ on the validation machine corresponding to $x$.

Next, we describe on which machine to schedule each clause job. Let $c$ be a clause in $\phi$. Let clause~$c$ be satisfied by its $\ell$th literal (if multiple literals satisfy the clause, pick one of them arbitrarily). Let $x$ be the variable appearing in the $\ell$th literal of $c$. We schedule clause job $c_\ell$ on the validation machine corresponding to $x$. We schedule clause jobs $c_{\ell'}$ with $\ell'\in\{1,2,3\}\setminus\{\ell\}$ on the two clause selection machines corresponding to $c$, respectively.

Lastly, notice that the number of dummy jobs equals the number of machines. For each dummy job we arbitrarily choose a distinct machine and schedule it on this machine.

In the constructed schedule, we clearly schedule each job. We next show that the schedule is feasible.

Notice that on each variable selection machine and each clause selection machine we schedule exactly two jobs, one dummy job and one variable job of the variable corresponding to the variable selection machine or, respectively, one clause job of the clause corresponding to the clause selection machine. Since the processing time intervals of the variable jobs of variables corresponding to the variable selection machines start at $2\alpha+2\beta$ and the deadline of each dummy job is at most $2\alpha+2\beta$, the schedules for the variable selection machines are feasible. Analogously, the schedules for the clause selection machines are feasible. 

It remains to show that the schedules for the validation machines are feasible. Notice that by construction, the variable jobs and clause jobs that are potentially scheduled on a validation machine cannot conflict with any dummy job. Furthermore, the variable jobs that are potentially scheduled on a validation machine cannot conflict with each other. We have the same for the clause jobs.

Hence, the only way to obtain an infeasible schedule is if a variable job and a clause job are in conflict. Assume the variable $x^T$ is scheduled (the case of variable job $x^F$ is analogous) and clause job $c_\ell$ is scheduled and the two jobs are in conflict. Note that this implies that we are dealing with the validation machine for variable $x$. By construction of the schedule, this means that variable $x$ is set to false in the satisfying assignment. However, the jobs $c_\ell$ and $x^T$ are in conflict on the validation machine for $x$ (and the job of $c_\ell$ is not in conflict with the dummy jobs) if $x$ appears non-negated in the $\ell$th literal of clause $c$. Furthermore, by construction of the schedule, we have that clause $c$ is satisfied by its $\ell$th literal. This is a contradiction to $x$ being set to false in the satisfying assignment.
\end{proof}

Now we show how to construct a satisfying assignment from a feasible schedule where all jobs are scheduled.

\begin{lemma}\label{lem:nphcorr2}
Let $\phi$ be an instance of \textsc{Exact (3,4)-SAT}. Let $I$ be the \ProbThree\ instance computed from $\phi$ as specified by \cref{const:nph}. If there is a feasible schedule $\sigma$ for $I$ such that all jobs are scheduled, then $\phi$ is satisfiable.
\end{lemma}

\begin{proof}
Let $\phi$ be an instance of \textsc{Exact (3,4)-SAT} and consider the corresponding \ProbThree\ instance specified by \cref{const:nph}.
Assume we have a feasible schedule for the constructed instance such that all jobs are scheduled. We construct a satisfying assignment for $\phi$ as follows. 

First, observe that by construction, the at most one dummy job can be scheduled on each machine. Since the number of dummy jobs equals the number of machines, we have that on each machine exactly one dummy job is scheduled. This means that on each machine no non-dummy jobs that conflict with a dummy job (that is, jobs with processing time equal to their deadline) can be scheduled.

We can further observe that on the variable selection machine of variable $x$, apart from a dummy job, only the variable job $x^T$ or the variable job $x^F$ can be scheduled. Since the two jobs conflict, they cannot both be scheduled. We assume w.l.o.g.\ that exactly one of the two jobs is scheduled. If the variable job $x^T$ is scheduled, we set variable $x$ to true, otherwise we set variable $x$ to false. In the remainder, we show that this yields a satisfying assignment for $\phi$.

Assume for contradiction that $\phi$ is not satisfied by the constructed assignment. Then there is a clause $c$ in $\phi$ such that none of its literals are satisfied. Consider the three clause jobs $c_1$, $c_2$, and $c_3$ associated with the three literals in clause $c$. Each of these three jobs can only be scheduled (without creating a conflict with a dummy job) on the clause selection machines corresponding to $c$, and the validation machine corresponding to the variable appearing in the respective literal of the clause $c$. Since we only have two clause selection machines, at least one  of the clause jobs $c_1$, $c_2$, and $c_3$ has to be scheduled on a validation machine. Assume $c_1$ is scheduled on a validation machine (the case of $c_2$ and $c_3$ is symmetric). Let $x$ be the variable appearing in the first literal of $c$. Assume the variable job $x^T$ is scheduled on the variable selection machine corresponding to $x$ (the case where the variable job $x^F$ is scheduled is symmetric). Then the variable job $x^F$ has to be scheduled on the validation machine corresponding to $x$, since on all other machines it is in conflict with all dummy jobs. However, by construction of the validation machines, the clause job $c_1$ and the variable job $x^F$ can only both be scheduled on the validation machine corresponding to $x$ if setting $x$ to true satisfies the first literal of $c$, a contradiction to the assumption that $c$ is not satisfied.
\end{proof}

Finally, we have all ingredients to prove \cref{thm:nph}.

\begin{proof}[Proof of \cref{thm:nph}]
To prove \cref{thm:nph}, we show that \cref{const:nph} is polynomial-time many-one reduction from \textsc{Exact (3,4)-SAT} to \ProbThree\. First, it is easy to observe that given an instance of \textsc{Exact (3,4)-SAT}, the \ProbThree\ instance specified by \cref{const:nph} can be computed in polynomial time. 
The correctness of the reduction follows from \cref{lem:nphcorr1,lem:nphcorr2}. Since \textsc{Exact (3,4)-SAT} is NP-hard~\cite{tovey1984simplified}, we have that \cref{thm:nph} follows.
\end{proof}

\section{Conclusion}\label{sec:3}

We proved that \ProbOne\ and its generalization \ProbTwo\ are W[1]-hard when parameterized by the number $m$ of machines, and that \ProbThree\ is NP-complete, answering two open questions by Mnich and van Bevern~\cite{mnich2018parameterized} and Sung and Vlach~\cite{sung2005maximizing}, respectively.
With this, we contribute to the understanding of the (parameterized) computational complexity of basic and natural interval scheduling problems. 

Our results leave two main open questions. Our NP-hardness proof for \ProbThree\ does not imply NP-hardness of \ProbFour, which leaves the following question.

\begin{openquestion}
What is the computational complexity of \ProbFour?
\end{openquestion}
Furthermore, the reduction used in our NP-hardness proof for \ProbThree\ uses an unbounded number of machines, and hence does not imply W[1]-hardness of \ProbThree\ when parameterized by the number $m$ of machines. Hence, we have the following question. 

\begin{openquestion}
What is the parameterized complexity of \ProbThree\ when parameterized by the number $m$ of machines?
\end{openquestion}

Lastly, we want to point out that the parameterized reduction we use to prove that \ProbOne\ and \ProbTwo\ are W[1]-hard when parameterized by the number $m$ of machines roughly squares the parameter. More precisely, we reduce from \textsc{Multicolored Clique} parameterized by the number $k$ of colors and for the number $m$ of machines in the produced \ProbOne\ / \ProbTwo\ instances we have $m\in O(k^2)$. This implies that assuming the Exponential Time Hypothesis (ETH)~\cite{IP01,IPZ01}, that there are no $f(m)(n+m)^{o(\sqrt{m})}$ algorithms for \ProbOne\ and \ProbTwo\ for any function $f$~\cite{lokshtanov2013lower}. However, the best known algorithms have running time $(n+m)^{O(m)}$~\cite{arkin1987scheduling,sung2005maximizing}. Hence, there is still a gap between the upper and lower bound.

\bibliography{literature}



\end{document}

%% file: figures/W1_edge_selection.tex
\def\eps{0.1}
\def\shiftx{5}

\colorlet{grayed}{gray!75}

    \draw[thick,->] (0,0) -- (17.5,0) node[anchor=north west] {};

    \draw[line width=1.5, |-|]
            ({0.5}, 1.2) -- (6.05, 1.2) node[anchor=north west] {};
    \node at (3, 1.6) {$j_v^{(\ell,\ell')}$};
            
    \draw[line width=1.5, |-|]
            ({\shiftx+6.6}, 1.2) -- (7+\shiftx+5, 1.2) node[anchor=north west] {};
    \node at (\shiftx+7+2.5, 1.6) {$j_w^{(\ell,\ell')}$};
     

    \draw[line width=1,grayed, |-|]
            (5, 1.5) -- (7.6, 1.5) node[anchor=north west] {};
        
     \node at (6, 1.9) {$j_v^{(\ell)}$};
\foreach \x\i\label in
    {
        0/1/$j_v^{(1)}$,
        1/2/$ $,
        2/3/$j_v^{(\ell')}$,
        3/4/$ $,
        4/5/$j_v^{(k)}$
    }
    {
        \ifthenelse{
        \(\i=2 \OR \i=4\)
        }
        {
            \node at (\eps+5+2*\x/5+\eps*\x+0.2, 0.3) {$\dots$};
            
        }
        {
        
            \ifthenelse{\(\NOT \i=3\)}
            {
            \draw[line width=1, grayed, |-|]
            (\eps+5+2*\x/5+\eps*\x, 0.3) -- (\eps+5+2*\x/5+0.4+\eps*\x, 0.3) node[anchor=north west] {};
            
            \node at (\eps+5+2*\x/5+\eps*\x+0.3, 0.7) {\label};
            }
            {
            \draw[line width=1, |-|]
            (\eps+5+2*\x/5+\eps*\x, 0.3) -- (\eps+5+2*\x/5+0.4+\eps*\x, 0.3) node[anchor=north west] {};
            
            \node at (\eps+5+2*\x/5+\eps*\x+0.3, 0.7) {\label};
            }
            
        }

    }

            
    \draw[line width=1,grayed, |-|]
            (5+\shiftx, 1.5) -- (7.5+\shiftx, 1.5) node[anchor=north west] {};
        
     \node at (6.5+\shiftx, 1.9) {$j_w^{(\ell')}$};
\foreach \x\i\label in
    {
        0/1/$j_w^{(1)}$,
        1/2/$ $,
        2/3/$j_w^{(\ell)}$,
        3/4/$ $,
        4/5/$j_w^{(k)}$
    }
    {
        
        \ifthenelse{
         \(\i=2 \OR \i=4\)
        }
        {
            \node at (\shiftx+\eps+5+2*\x/5+\eps*\x+0.2, 0.3) {$\dots$};
            
        }
        {
             \ifthenelse{ \(\not \i=3\)}
             {
                \draw[line width=1, grayed, |-|]
                (\shiftx+\eps+5+2*\x/5+\eps*\x, 0.3) -- (\shiftx+\eps+5+2*\x/5+0.4+\eps*\x, 0.3) node[anchor=north west] {};
                
                \node at (\shiftx+\eps+5+2*\x/5+\eps*\x+0.3, 0.7) {\label};
             }
             {
                \draw[line width=1, |-|]
                (\shiftx+\eps+5+2*\x/5+\eps*\x, 0.3) -- (\shiftx+\eps+5+2*\x/5+0.4+\eps*\x, 0.3) node[anchor=north west] {};
            
            \node at (\shiftx+\eps+5+2*\x/5+\eps*\x+0.3, 0.7) {\label};
             }
            
        }
        
    }
    
    
    \draw[line width=1, |-|]
            (6.5, 2) -- (11., 2) node[anchor=north west] {};
    \node at (8.5, 2.4) {$j_{\{u,w\}}$};

%% file: figures/W1_validation_machine.tex
\def\eps{0.1}
\def\shiftx{5}

    \draw[thick,->] (-2,0) -- (14,0) node[anchor=north west] {};

\node at (\shiftx*0.5+6.3, 1.2) {$\dots$};
\node at (\shiftx*-0.5+6.3, 1.2) {$\dots$};
 
\foreach \j\v in {-1/v,0/u,1/w}
{

    \draw[line width=1, |-|]
            (\shiftx*\j+5, 1.5) -- (\shiftx*\j+7.6, 1.5) node[anchor=north west] {};
    \node at (\shiftx*\j+6.3, 1.9) {$j_\v^{(\ell)}$};
    \foreach \x\i\label in
    {
        0/1/$j_\v^{(1)}$,
        1/2/$ $,
        2/3/$j_\v^{(\ell')}$,
        3/4/$ $,
        4/5/$j_\v^{(k)}$
    }
    {
        
        \ifthenelse{ \(\i=2 \OR \i=4\)}
        {
            \node at (\shiftx*\j+\eps+5+2*\x/5+\eps*\x+0.2, 0.4) {$\dots$};
            
        }
        {
            \draw[line width=1, |-|]
            (\shiftx*\j+\eps+5+2*\x/5+\eps*\x, 0.4) -- (\shiftx*\j+\eps+5+2*\x/5+0.4+\eps*\x, 0.4) node[anchor=north west] {};
            
            \node at (\shiftx*\j+\eps+5+2*\x/5+\eps*\x+0.3, 0.8) {\label};
        }
        
    }

}

  

%% file: figures/variable_selection.tex
  \def\dx{2}
  \def\shiftx{1.5}
  \def\shifty{2.5}
  \def\eps{0.35}
  \def\shift{0.3}

  \draw[thick,->] (0,1.5) -- (17.5,1.5) node[anchor=north west] {};

\foreach \x\i\name\expression in
    {
        1/1/$D$/,
        2/2/$N$/,
        3/3/$F$/$F$,
        4/4/$P$/,
        5/5/$T$/$N$
    }
    {

      {

        \draw[line width=1, |-|]
        ({1}, 2.+\x*\shifty+\eps*\x) -- ({3+\x+(\x-1)*\dx}, 2.+\x*\shifty+\eps*\x) node[anchor=north west] {};

        \ifthenelse{
        \(\x=5 \OR \x=3\)
        }
        {
          \node at ({2+\x+(\x-1)*\dx}, 2.+\x*\shifty+\eps*\x-0.25*\shifty) {\vdots};

          \ifthenelse{\x=5}{
            \draw[line width=2, |-|]
            ({3+1+(1-0.75)*\dx+0.2}, 2.+\x*\shifty+\eps*\x-0.5*\shifty) -- ({3+\x+(\x-1.)*\dx-\shiftx*0.5}, 2.+\x*\shifty+\eps*\x-0.5*\shifty) node[anchor=east] {};
            \node at ({3+\x+(\x-1.)*\dx-\shiftx*0.5+0.5}, 2.+\x*\shifty+\eps*\x-0.5*\shifty) {$x^T_i$};
          }
          {
            \draw[line width=2, |-|]
            ({3+1+(1-0.75)*\dx+0.2}, 2.+\x*\shifty+\eps*\x-0.5*\shifty) -- ({3+\x+(\x-1.)*\dx-\shiftx*0.5}, 2.+\x*\shifty+\eps*\x-0.5*\shifty) node[anchor=east] {};
            \node at ({3+\x+(\x-1.)*\dx-\shiftx*0.5+0.5}, 2.+\x*\shifty+\eps*\x-0.5*\shifty) {$x^F_j$};
          }

          \node at ({2+\x+(\x-1.25)*\dx}, 2.+\x*\shifty+\eps*\x-0.75*\shifty) {\vdots};
        }
        {
        \node at ({2+\x+(\x-1)*\dx}, 2.+\x*\shifty+\eps*\x-0.25*\shifty) {\vdots};

        \draw[line width=1, |-|]
        ({1}, 2.+\x*\shifty+\eps*\x-0.5*\shifty) -- ({3+\x+(\x-1.)*\dx-\shiftx*0.5}, 2.+\x*\shifty+\eps*\x-0.5*\shifty) node[anchor=east] {};
        \node at ({3+\x+(\x-1.)*\dx-\shiftx*0.5+0.5+0.13}, 2.+\x*\shifty+\eps*\x-0.5*\shifty) {\expression};

        \node at ({2+\x+(\x-1.25)*\dx}, 2.+\x*\shifty+\eps*\x-0.75*\shifty) {\vdots};
        }

        \draw[line width=1, |-|]
        ({1}, 2.+\x*\shifty-\shifty+\eps*\x) -- ({3+\x+(\x-1)*\dx-\shiftx}, 2.+\x*\shifty-\shifty+\eps*\x) node[anchor=north west] {};
        
         \ifthenelse{ \not \x=5}{
             \draw[line width=1, dash dot]
            ({3+\x+(\x-0.75)*\dx+0.1}, 1.5) -- ({3+\x+(\x-0.75)*\dx+0.1}, 18) node[anchor=north west] {};
            \node at ({3+\x+(\x-1.5)*\dx}, 17) {\large \name};
         }{
         \node at ({3+\x+(\x-1.75)*\dx}, 17) {\large \name};
         }

      }
    }

%% file: figures/clause_selection.tex
  \def\dx{2}
  \def\shiftx{1.5}
  \def\shifty{2.5}
  \def\eps{0.35}
  \def\shift{0.3}

  \draw[thick,->] (0,1.5) -- (17.5,1.5) node[anchor=north west] {};

\foreach \x\i\name\expression in
    {
        1/1/$D$/,
        2/2/$N$/$c \in N$,
        3/3/$F$/,
        4/4/$P$/$c \in P$,
        5/5/$T$/
    }
    {

      {

        \draw[line width=1, |-|]
        ({1}, 2.+\x*\shifty+\eps*\x) -- ({3+\x+(\x-1)*\dx}, 2.+\x*\shifty+\eps*\x) node[anchor=north west] {};

        \ifthenelse{ \(\x=4 \OR \x=2 \)}
        {

          \ifthenelse{\x=4}{
            \draw[line width=2, |-|]
            ({3+1+(1-0.75)*\dx+0.2}, 2.+\x*\shifty+\eps*\x-0.7*\shifty) -- ({3+\x+(\x-1.)*\dx-\shiftx*0.66}, 2.+\x*\shifty+\eps*\x-0.7*\shifty) node[anchor=east] {};
            \node at ({3+\x+(\x-1.)*\dx-\shiftx*0.33+0.5},  2.+\x*\shifty+\eps*\x-0.3*\shifty) {$c_1$};

            \node at ({3+\x+(\x-1.)*\dx-\shiftx*1.1+0.5}, 2.+\x*\shifty+\eps*\x-0.5*\shifty) {$\vdots$};
            
            \node at ({3+\x+(\x-1.)*\dx-\shiftx*0.9+0.5}, 2.+\x*\shifty+\eps*\x-0.1*\shifty) {$\vdots$};
            
            \node at ({3+\x+(\x-1.)*\dx-\shiftx*1.3+0.5}, 2.+\x*\shifty+\eps*\x-0.8*\shifty) {$\vdots$};
            
            \draw[line width=2, |-|]
            ({3+1+(1-0.75)*\dx+0.2}, 2.+\x*\shifty+\eps*\x-0.3*\shifty) -- ({3+\x+(\x-1.)*\dx-\shiftx*0.33}, 2.+\x*\shifty+\eps*\x-0.3*\shifty) node[anchor=east] {};
           
             \node at ({3+\x+(\x-1.)*\dx-\shiftx*0.66+0.5}, 2.+\x*\shifty+\eps*\x-0.7*\shifty) {$c_2$};
            
          }
          {
          
            \node at ({2+\x+(\x-1)*\dx}, 2.+\x*\shifty+\eps*\x-0.25*\shifty) {\vdots};
            
            \draw[line width=2, |-|]
            ({3+1+(1-0.75)*\dx+0.2}, 2.+\x*\shifty+\eps*\x-0.5*\shifty) -- ({3+\x+(\x-1.)*\dx-\shiftx*0.5}, 2.+\x*\shifty+\eps*\x-0.5*\shifty) node[anchor=east] {};
            \node at ({3+\x+(\x-1.)*\dx-\shiftx*0.5+0.5}, 2.+\x*\shifty+\eps*\x-0.5*\shifty) {$c_3$};
            
            \node at ({2+\x+(\x-1.25)*\dx}, 2.+\x*\shifty+\eps*\x-0.75*\shifty) {\vdots};
          }

        }
        {
        \node at ({2+\x+(\x-1)*\dx}, 2.+\x*\shifty+\eps*\x-0.25*\shifty) {\vdots};

        \draw[line width=1, |-|]
        ({1}, 2.+\x*\shifty+\eps*\x-0.5*\shifty) -- ({3+\x+(\x-1.)*\dx-\shiftx*0.5}, 2.+\x*\shifty+\eps*\x-0.5*\shifty) node[anchor=east] {};
        \node at ({3+\x+(\x-1.)*\dx-\shiftx*0.5+0.5+0.13}, 2.+\x*\shifty+\eps*\x-0.5*\shifty) {\expression};

        \node at ({2+\x+(\x-1.25)*\dx}, 2.+\x*\shifty+\eps*\x-0.75*\shifty) {\vdots};
        }

        \draw[line width=1, |-|]
        ({1}, 2.+\x*\shifty-\shifty+\eps*\x) -- ({3+\x+(\x-1)*\dx-\shiftx}, 2.+\x*\shifty-\shifty+\eps*\x) node[anchor=north west] {};
        
         \ifthenelse{ \not \x=5}{
             \draw[line width=1, dash dot]
            ({3+\x+(\x-0.75)*\dx+0.1}, 1.5) -- ({3+\x+(\x-0.75)*\dx+0.1}, 18) node[anchor=north west] {};
            \node at ({3+\x+(\x-1.5)*\dx}, 17) {\large \name};
         }{
         \node at ({3+\x+(\x-1.75)*\dx}, 17) {\large \name};
         }

      }
    }

%% file: figures/validation_machine.tex
  \def\dx{2}
  \def\shiftx{1.5}
  \def\shifty{2.5}
  \def\eps{0.35}
  \def\shift{0.3}

  \draw[thick,->] (0,1.5) -- (17.5,1.5) node[anchor=north west] {};

\foreach \x\i\name\expression in
    {
        1/1/$D$/,
        2/2/$N$/,
        3/3/$F$/,
        4/4/$P$/,
        5/5/$T$/
    }
    {

      {

        \draw[line width=1, |-|]
        ({1}, 2.+\x*\shifty+\eps*\x) -- ({3+\x+(\x-1)*\dx}, 2.+\x*\shifty+\eps*\x) node[anchor=north west] {};

        \ifthenelse{ \(\x=5 \OR \x=3 \)}
        {
          \node at ({2+\x+(\x-1)*\dx}, 2.+\x*\shifty+\eps*\x-0.25*\shifty) {\vdots};

          \ifthenelse{\x=5}{
            \draw[line width=2, |-|]
            ({3+3+(3-1.)*\dx-\shiftx*0.5+0.25}, 2.+\x*\shifty+\eps*\x-0.5*\shifty) -- ({3+\x+(\x-1.)*\dx-\shiftx*0.5}, 2.+\x*\shifty+\eps*\x-0.5*\shifty) node[anchor=east] {};
            \node at ({3+\x+(\x-1.)*\dx-\shiftx*0.5+0.5}, 2.+\x*\shifty+\eps*\x-0.5*\shifty) {$x^T$};
          }
          {
            \draw[line width=2, |-|]
            ({3+1+(1-0.75)*\dx+0.2}, 2.+\x*\shifty+\eps*\x-0.5*\shifty) -- ({3+\x+(\x-1.)*\dx-\shiftx*0.5}, 2.+\x*\shifty+\eps*\x-0.5*\shifty) node[anchor=east] {};
            \node at ({3+\x+(\x-1.)*\dx-\shiftx*0.5+0.5}, 2.+\x*\shifty+\eps*\x-0.5*\shifty) {$x^F$};
          }

          \node at ({2+\x+(\x-1.25)*\dx}, 2.+\x*\shifty+\eps*\x-0.75*\shifty) {\vdots};
        }
        {

            \ifthenelse{\x=1}{
                
                 \node at ({2+\x+(\x-1)*\dx}, 2.+\x*\shifty+\eps*\x-0.25*\shifty) {\vdots};
    
                \draw[line width=1, |-|]
                ({1}, 2.+\x*\shifty+\eps*\x-0.5*\shifty) -- ({3+\x+(\x-1.)*\dx-\shiftx*0.5}, 2.+\x*\shifty+\eps*\x-0.5*\shifty) node[anchor=east] {};
                \node at ({3+\x+(\x-1.)*\dx-\shiftx*0.5+0.5+0.13}, 2.+\x*\shifty+\eps*\x-0.5*\shifty) {\expression};

                \node at ({2+\x+(\x-1.25)*\dx}, 2.+\x*\shifty+\eps*\x-0.75*\shifty) {\vdots};
            }
            {   
                \ifthenelse{\not \x=2}{
                \draw[line width=2, |-|]
                ({3+\x+(\x-1.)*\dx-\shiftx*0.66 -0.5}, 2.+\x*\shifty+\eps*\x-0.7*\shifty) -- ({3+\x+(\x-1.)*\dx-\shiftx*0.66}, 2.+\x*\shifty+\eps*\x-0.7*\shifty) node[anchor=east] {};
                \node at ({3+\x+(\x-1.)*\dx-\shiftx*0.33+0.5},  2.+\x*\shifty+\eps*\x-0.3*\shifty) {$c_1$};

                \node at ({3+\x+(\x-1.)*\dx-\shiftx*1.1+0.5}, 2.+\x*\shifty+\eps*\x-0.5*\shifty) {$\vdots$};
                
                \node at ({3+\x+(\x-1.)*\dx-\shiftx*0.9+0.5}, 2.+\x*\shifty+\eps*\x-0.1*\shifty) {$\vdots$};
                
                \node at ({3+\x+(\x-1.)*\dx-\shiftx*1.3+0.5}, 2.+\x*\shifty+\eps*\x-0.8*\shifty) {$\vdots$};
                
                \draw[line width=2, |-|]
                ({3+\x+(\x-1.)*\dx-\shiftx*0.33-0.5}, 2.+\x*\shifty+\eps*\x-0.3*\shifty) -- ({3+\x+(\x-1.)*\dx-\shiftx*0.33}, 2.+\x*\shifty+\eps*\x-0.3*\shifty) node[anchor=east] {};
               
                 \node at ({3+\x+(\x-1.)*\dx-\shiftx*0.66+0.5}, 2.+\x*\shifty+\eps*\x-0.7*\shifty) {$c'_2$};
                }
                {
                \draw[line width=2, |-|]
                ({3+\x+(\x-1.)*\dx-\shiftx*0.66 -0.5}, 2.+\x*\shifty+\eps*\x-0.7*\shifty) -- ({3+\x+(\x-1.)*\dx-\shiftx*0.66}, 2.+\x*\shifty+\eps*\x-0.7*\shifty) node[anchor=east] {};
                \node at ({3+\x+(\x-1.)*\dx-\shiftx*0.33+0.5},  2.+\x*\shifty+\eps*\x-0.3*\shifty) {$c''_3$};

                \node at ({3+\x+(\x-1.)*\dx-\shiftx*1.1+0.5}, 2.+\x*\shifty+\eps*\x-0.5*\shifty) {$\vdots$};
                
                \node at ({3+\x+(\x-1.)*\dx-\shiftx*0.9+0.5}, 2.+\x*\shifty+\eps*\x-0.1*\shifty) {$\vdots$};
                
                \node at ({3+\x+(\x-1.)*\dx-\shiftx*1.3+0.5}, 2.+\x*\shifty+\eps*\x-0.8*\shifty) {$\vdots$};
                
                \draw[line width=2, |-|]
                ({3+\x+(\x-1.)*\dx-\shiftx*0.33-0.5}, 2.+\x*\shifty+\eps*\x-0.3*\shifty) -- ({3+\x+(\x-1.)*\dx-\shiftx*0.33}, 2.+\x*\shifty+\eps*\x-0.3*\shifty) node[anchor=east] {};
               
                 \node at ({3+\x+(\x-1.)*\dx-\shiftx*0.66+0.5}, 2.+\x*\shifty+\eps*\x-0.7*\shifty) {$c'''_1$};
                }

            }

        }

        \draw[line width=1, |-|]
        ({1}, 2.+\x*\shifty-\shifty+\eps*\x) -- ({3+\x+(\x-1)*\dx-\shiftx}, 2.+\x*\shifty-\shifty+\eps*\x) node[anchor=north west] {};

        \ifthenelse{ \not \x=5}{
            \draw[line width=1, dash dot]
            ({3+\x+(\x-0.75)*\dx+0.1}, 1.5) -- ({3+\x+(\x-0.75)*\dx+0.1}, 18) node[anchor=north west] {};
            \node at ({3+\x+(\x-1.5)*\dx}, 17) {\large \name};
        }{
            \node at ({3+\x+(\x-1.75)*\dx}, 17) {\large \name};
        }

      }
    }